\documentclass[journal]{IEEEtran}
\usepackage{ifpdf}
\usepackage{cite}
\usepackage{algorithmic}
\usepackage{algorithm}
\ifCLASSINFOpdf
   \usepackage[pdftex]{graphicx}

\else
 \usepackage[dvips]{graphicx}
\fi

\ifCLASSOPTIONcompsoc
    \usepackage[caption=false, font=normalsize, labelfont=sf, textfont=sf]{subfig}
\else
\usepackage[caption=false, font=footnotesize]{subfig}
\fi
\usepackage[cmex10]{amsmath}
\usepackage{epstopdf}
\usepackage{array} 
\usepackage{flushend}
\usepackage{balance}
\usepackage{eqparbox}

\usepackage{fixltx2e}
\usepackage{color}
\usepackage{float}

\usepackage{stfloats}
\usepackage{url}
\usepackage{cite}
\usepackage{amsthm}
\newtheorem{theorem}{Theorem}

\newtheorem{lemma}[theorem]{Lemma}
\newcommand*{\Scale}[2][4]{\scalebox{#1}{$#2$}}
\usepackage{amssymb}
\begin{document}

\title{Power Allocation and Link Selection for Multicell Cooperative NOMA Hybrid VLC/RF Systems }

\setlength{\columnsep}{0.21 in}

\author{Mohanad~Obeed,
        Hayssam~Dahrouj,~\IEEEmembership{Senior Member,~IEEE,}
        Anas~M.~Salhab,~\IEEEmembership{Senior Member,~IEEE,}
        Anas~Chaaban,~\IEEEmembership{Member,~IEEE,}
        Salam~A.~Zummo,~\IEEEmembership{Senior Member,~IEEE,} and
        Mohamed-Slim~Alouini,~\IEEEmembership{Fellow,~IEEE}
        
}
\maketitle

\begin{abstract}
 This paper proposes and optimizes a cooperative non-orthogonal multiple-access (Co-NOMA) scheme in the context of multicell visible light communications (VLC) networks, as a means to mitigate inter-cell interference in Co-NOMA-enabled systems. Consider a network with multiple VLC access points (APs), where each AP serves two users using light intensity. In each cell, the weak user (the cell edge user) can be served either directly by the VLC AP, or through the strong user that can decode the weak user's message and forward it through the radio-frequency (RF) link. The paper then considers the problem of maximizing the network throughput under quality-of-service (QoS) constraints by allocating the powers of the users' messages and APs' transmit powers, and determining the serving links of each weak user (i.e., VLC or hybrid VLC/RF). The paper solves such a non-convex problem by first finding closed form solutions of the joint users' powers and link selection for a fixed AP power allocation. The APs' transmit powers are then iteratively solved in an outer loop using the golden section method. Simulation results show how the proposed solution and scheme improve the system sum-rate and fairness as compared to conventional non-orthogonal multiple-access (NOMA) schemes.

\end{abstract}
\IEEEpeerreviewmaketitle
\section{Introduction}
The recent escalating need for high data rates and the increasing number of connected devices has necessitated a thorough examination of the vast, unregulated, free visible light spectrum through visible light communications (VLC). The performance of multicell VLC systems is, however, limited by inter-cell interference, as the lamps used as transmitters in indoor environments are often mounted close to each other to achieve sufficient illumination levels \cite{Obeed2018}. This paper considers a cooperative non-orthogonal multiple-access (Co-NOMA) scheme with multiple VLC access points (APs), where each AP serves two users using light intensity. In each cell, the weak user can be served either directly by the VLC AP, or through the strong user that can decode the weak user's message and forward it through the radio-frequency (RF) link. The paper then focuses on maximizing the throughput, which is a function of the levels of allocated powers and links serving each weak user (i.e., VLC or hybrid VLC/RF).

The scheme proposed in this paper is related to the recent literature on interference management in VLC systems \cite{coop,ffr1l,hanzo_haas,TWC}. The authors in \cite{coop} show that supporting VLC networks by RF APs would mitigate the effect of interference. Another approach used to mitigate interference is through fractional frequency reuse \cite{ffr1l}. References \cite{hanzo_haas} and \cite{TWC} use the joint transmission and user-centric design to cancel or decrease the interference levels. Moreover, similar to the non-orthogonal multiple-access (NOMA) scheme applied in classical RF networks \cite{ding2017survey}, recent works \cite{kizilirmak2015non, yin2016performance, zhang2017user} apply the NOMA principle to VLC networks and show that NOMA scheme outperforms the orthogonal-frequency division multiple-access (OFDMA) scheme \cite{kizilirmak2015non}, and the orthogonal multiple-access (OMA) scheme \cite{ yin2016performance}. Reference \cite{zhang2017user}, on the other hand, investigates the advantages harvested by user grouping and power allocation for NOMA-enabled VLC systems using NOMA. As a means to further improve wireless systems performance, Co-NOMA has been recently proposed to exploit the redundant information to strengthen the received signal-to-noise-ratio (SNR) at the weak receivers in RF networks \cite{Liu2016}, and in VLC systems \cite{CoNOMA}. All the aforementioned references, however, focus on a single-cell case and ignore the potential inter-cell interference, unlike our paper which considers a multi-cell VLC system scenario.

The paper considers a multi-cell Co-NOMA VLC system, where each AP serves two users using light intensity. In each cell, the weak user can be served either directly by the VLC AP, or through the strong user that can decode the weak user's message and forward it through the RF link. The paper then formulates the problem of maximizing the sum-rate under QoS, APs' transmit power, and connectivity constraints so as to allocate the powers of the  users' messages, link selection vector, and the transmit power of each AP. The paper tackles such a difficult non-convex optimization problem iteratively by first finding closed-form solutions of the users' powers and link selection vector problem, under fixed AP power scenario. We then find a solution for the AP transmit power in an outer loop using the golden section method. Simulation results show that the proposed solution and scheme outperform the NOMA scheme in terms of sum-rate and fairness.
\section{System Model and Problem Formulation}
\begin{figure} [!t]
\centering
\includegraphics[width=3in]{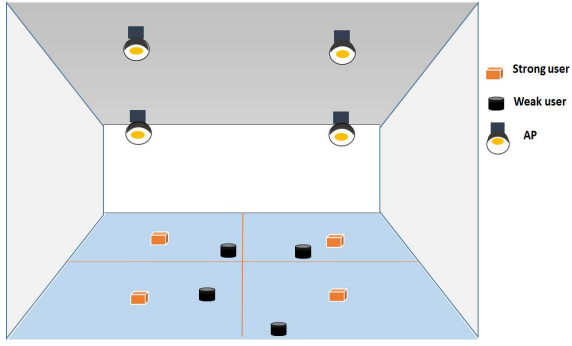}
\caption{An example of system model with 4 cells.}
\label{SMk}
\end{figure}
%\begin{figure}
%\label{SM}
%    \centering
%    \label{SM1}
%  \subfloat[An example of system model with 4 APs]{\includegraphics[trim=0.5cm 7cm 0cm 12cm width=0.2\linewidth, height=3.6cm]{SM77.eps}}
%
%  \subfloat[Distribution areas of the strong and weak users]{
%        \includegraphics[trim=0cm 3cm 0cm 7cm width=0.4\linewidth, height=3.6cm]{Cell7.eps}}
%    \label{SM2}
%    \caption{System Model}
%\end{figure}
%\begin{figure}[!t]
%    \centering
%  \includegraphics[trim=-2cm 6cm 16cm 11cm width=0.7\linewidth, height=3.8cm]{SM_IPE.eps}
%  %\includegraphics[width=3in]{SM77.eps}
%    \caption{An example of system model with 4 cells.}
%    \label{SMk}
%\end{figure}
%\begin{figure}[!t]
%\centering
% \includegraphics[ width=0.5\linewidth, height=3.8cm]{Cell2}
%    \caption{Distribution areas of the strong and weak users.}
%    \label{SM2}
%\end{figure}
\subsection{System Model}
The paper considers a system model, consisting of $N$ APs fixed at the ceiling, each serving two users as shown in Fig. \ref{SMk}. We assume that each user is served by the closest AP. In each cell, the two users are distributed in a way that one of them (strong user) is around the cell center and the other (weak user) is located near the cell edge. We define a parameter $\alpha$ that indicates the average interference received by the weak users. Increasing $\alpha$ limits the distribution area of the weak users to be closer to the cell edge.
 %within a circle around the center cell and the other (weak user) is located outside a square closer to the cell-edge. Increasing the parameter $\alpha$ in Fig. \ref{SM2} (with fixing circle A) shrinks the area of the weak users to be closer to the cell-edge, which increases the average received interference at the weak users.
 All the VLC APs in the system share the whole available VLC bandwidth, which leads to inter-cell interference. We assume that the strong user can also act as a relay, which can decode the weak user's message, harvest the energy through the received light intensity, and then use this energy to forward the decoded message to the weak user using the RF link. The system, therefore, can serve the weak users either by the VLC AP directly, or by the paired strong user through the hybrid VLC/RF link. The transmitted signal from the AP $k$ to the weak user $w$ and the strong user $s$ is  $y_k=\nu\sqrt{P_{s,k}}s_{s,k}+\nu\sqrt{P_{w,k}}s_{w,k}+\nu b_k,$
where $P_{s,k}$ and $P_{w,k}$ are the powers assigned to the strong and weak users' symbols $s_{s,k}$ and $s_{w,k}$, respectively, $b$ is the direct-current (DC) that must be added to guarantee that the transmitted signal is non-negative, where $E[\vert s_{s,k}\vert^2]=1$, $E[\vert s_{w,k}\vert^2]=1$, and $\nu$ is the proportionality factor of the electric-to-optical power conversion. Because we have $N$ APs transmitting data using the visible light intensity, the received signals at the strong user associated to the $k$th AP is given by
\begin{multline}
z_{s,k}=\rho\nu h_{s,k}\sqrt{P_{s,k}}s_{s,k}+\rho\nu h_{s,k}\sqrt{P_{w,k}}s_{w,k}\\
\Scale[0.9]{+\rho\nu\sum_{q=1,q\neq k}^N \bigg( h_{s,q}\sqrt{P_{s,q}}s_{s,q}+h_{s,q}\sqrt{P_{w,q}}s_{w,q}+b\bigg)+n,}
\end{multline} where $\rho$ is the optical-to-electrical conversion factor, $h_{s,k}$ is the VLC channel between the AP $k$ and the strong user in cell $k$, and $n$ is the noise components which can be modeled as real zero-mean additive white Gaussian noise (AWGN) variable with variance $\sigma^2 = N_vB_v$, where $N_v$ is the noise power spectral density (PSD) and $B_v$ is the modulation bandwidth. Each strong user then decodes the weak user's signal and uses the received DC signal to harvest energy. The harvested energy then can be used to forward the weak user's signal through the RF link \cite{CoNOMA}. Therefore, the achievable rate of the strong user of the $k$th cell can be lower bounded by \cite{chaaban2016capacity}
\begin{equation}
\label{Rsk}
\Scale[1.08]{R_{s,k}=\frac{B_v}{2}\log_2\left( 1+\frac{c\nu^2\rho^2 h_{s,k}^2P_{s,k}}{B_v N_v+c\nu^2\rho^2\sum_{q=1,q\neq k}^{N_A} p_qh_{s,q}^2}\right)},
\end{equation}
where $p_q=P_{w,q}+P_{s,q}$ is the transmit power at the AP $q$, $c=\min\{\frac{1}{2\pi e}, \frac{e b^2}{I_H^22\pi}\}$, $e$ is Euler's number, and $I_H$ is the allowable maximum input current to the LED.
The achievable data rate of the weak user, served by the $k$th AP through the direct VLC link, is given by
\begin{multline}
\label{Rwk}
R_{w,k}=\frac{B_v}{2}\log_2\bigg( 1+\\
\frac{c\nu^2\rho^2 h_{w,k}^2P_{w,k}}{B_v N_v+c\nu^2\rho^2P_{s,k}h_{w,k}^2+c\nu^2\rho^2\sum_{q=1,q\neq k}^{N_A} p_qh_{w,q}^2}\bigg),
\end{multline}
while the achievable data rate of the weak user, served by the $k$th AP through the hybrid VLC/RF link is given by
\begin{equation}
R_{w,k}^{HL}=\min \left(R_{w\rightarrow s,k},R_{w,s,k}^{RF} \right),
\end{equation}
where \begin{multline}
R_{w\rightarrow s,k}=\frac{B_v}{2}\log_2\bigg( 1+\\
\frac{c\nu^2\rho^2 h_{s,k}^2P_{w,k}}{B_v N_v+c\nu^2\rho^2P_{s,k}h_{s,k}^2+c\nu^2\rho^2\sum_{q=1,q\neq k}^{N_A} p_qh_{s,q}^2}\bigg)
\end{multline}
 is the achievable data rate of the weak user received at the strong user and
 \begin{equation}
 R_{w,s,k}^{RF}=B_{RF}\log_2\left(1+\frac{P_{RF,s,k}h_{RF,w,s,k}^2}{B_{RF}N_{RF}}\right)
\end{equation}
  is the achievable data rate of the weak user that can be provided by the strong user through the RF link, where $P_{RF,s,k}$ is the harvested power at the strong user in the cell $k$, which depends on the DC biases at all the APs \cite{obeed2018dc}, $h_{RF,w,s,k}$ is the RF channel between the strong user and the weak user in cell $k$, $B_{RF}$ is the RF bandwidth assigned for one user, and $N_{RF}$ is the RF noise power spectral density.
\subsection{Problem Formulation}
Our goal in this paper is to maximize the sum-rate of the system under QoS constraints, user connectivity, and maximum transmit power constraints by finding the powers of the users' messages and the link selection vector of the weak users. Define $\mathbf{x}=[x_1,x_2,\ldots,x_N]$ as the link selection vector, where $x_k=1$ means that the weak user in cell $k$ is served through the VLC/RF link, and $x_k=0$ means that the weak user is served directly through the VLC link from AP $k$.
The problem, then, can be formulated as follows
\begin{subequations}
\label{SRT}
\begin{eqnarray}
&\displaystyle\max_{\mathbf{p}, \mathbf{P}, \mathbf{x}}& \sum_{k=1}^N \left( R_{s,k}+(1-x_{k})R_{w,k}+x_{k}R_{w,k}^{HL}\right)\\
\label{SRTb}
&\text{s.t.}&   R_{s,k}\geq R_{th},\ \forall k\\
\label{SRTc}
&& (1-x_{k})R_{w,k}+x_{k}R_{w,k}^{HL}\geq R_{th},\ \forall k\\
\label{SRTd}
&&   P_{w,k}+P_{s,k}\leq p_k,\ \forall k=1,\ldots N\\
\label{SRTg}
&& x_{k}\in \lbrace 0,1  \rbrace,\  \forall k=1,\ldots N,
%\label{SRTh}
%&& 0\leq P_{s,k}\leq P_{w,k}, \  \forall k=1,\ldots N,
\end{eqnarray}
\end{subequations}
where $p_k$ is the total transmit power of the $k$th AP, i.e., the $k$th entry of the power vector defined by $\mathbf{p}$. Constraints (\ref{SRTb}) and (\ref{SRTc}) are imposed to guarantee the required QoS for the strong and weak users, respectively. Constraints (\ref{SRTd}) denote the transmit power constraints.  Constraint (\ref{SRTg}) is imposed to guarantee that each weak user is either connected directly to the VLC AP, or through the strong user by using the hybrid VLC/RF link.   Problem (\ref{SRT}) is a mixed-integer non-convex optimization problem with a non-concave objective function. The paper, therefore, proposes solving the problem in an iterative way, where the users' powers and link selection vector problem are first solved under fixed AP power scenario. We then find a solution for the AP transmit power in an outer loop via the golden section method. The simulations illustrate the efficiency of the proposed solution as compared to NOMA.
\section{Joint Power Allocation and Link Selection}
To solve problem (\ref{SRT}), we first find the optimal joint solution of the users' power and link selection vectors, for a fixed transmit power vector $\mathbf{p}$. In the outer loop, we then solve the transmit power so as to maximize the objective function in (\ref{SRT}). For a fixed power vector $\mathbf{p}$, problem (\ref{SRT}) can be equivalently divided into $N$ problems, where each problem can be solved at the corresponding AP. Hence, the problem at AP $k$ can be formulated as follows
\begin{subequations}
\label{SRT1}
\begin{eqnarray}
&\displaystyle\max_{ P_{s,k}, P_{w,k}, \mathbf{x}}&  R_{s,k}+(1-x_{k})R_{w,k}+x_{k}R_{w,k}^{HL}\\
\label{SRT1b}
&\text{s.t.}&   R_{s,k}\geq R_{th},\ \\
\label{SRT1c}
&& (1-x_{k})R_{w,k}+x_{k}R_{w,k}^{HL}\geq R_{th}, \\
\label{SRT1d}
&&   P_{w,k}+P_{s,k}\leq p_k, \\
\label{SRT1g}
&& x_{k}\in \lbrace 0,1  \rbrace,  \\
\label{SRT1h}
&& 0\leq P_{s,k}\leq P_{w,k}.
\end{eqnarray}
\end{subequations}
%\begin{subequations}
%\label{SR}
%\begin{eqnarray}
%\nonumber
%&\displaystyle\max_{\mathbf{p}, \mathbf{P}, \mathbf{x}}&  R_{j,k}+(1-x_{k})+(1-x_{k})R_{i,k}+x_{k}\min \min \left(R_{i\rightarrow j,k},R_{i,j}^{RF} \right)\\
%\label{SRb}
%&\text{s.t.}&   R_{k,j}\geq R_{th}\\
%\label{SRc}
%&& (1-x_{k})R_{k,i}+x_{k}\min(R_{k,j\rightarrow i},R_{k,i}^{RF})\geq R_{th}\\
%\label{SRd}
%&&   P_{k,i}+P_{k,j}\leq p_k\\
%\label{SRe}
%&& p_k=(I_H-b_k)^2\\
%\label{SRf}
%&& P_{rf,j}=f \rho \nu V_t \left(b_k h_{k,j}+\sum_{q=1,q\neq k}^{N_A}b_qh_{q,j}\right)\ln\left(1+\frac{\rho \nu \left(b_k h_{k,j}+\sum_{q=1,q\neq k}^{N_A}b_qh_{q,j}\right)}{I_0}\right)\\
%\label{SRg}
%&& x_{k}\in \lbrace 0,1  \rbrace,\\
%\label{SRh}
%&& 0\leq P_{k,j}\leq P_{k,i},\\
%\label{SRi}
%&&\frac{I_H+I_L}{2}\leq b\leq I_H,
%\end{eqnarray}
%\end{subequations}
Problem (\ref{SRT1}) is still not convex, because of the binary variables and the interference terms in the expression of $R_{w,k}$. In the following, however, we provide a closed-form solution for the cases $x_k=1$ and $x_k=0$, respectively. First, we note that constraint (\ref{SRT1d}) should be satisfied with equality at the optimal solution since, otherwise, one can increase $P_{w,k}$ to achieve the equality, which increases the objective function without violating the constraints. This means that problem (\ref{SRT1}) can be expressed in terms of $P_{s,k}$ and $P_{w,k}$, which can be found by plugging $P_{w,k}=p_k-P_{s,k}$ in (\ref{SRT1}). Now, we discuss the solution for the two cases $x_k=0$ and $x_k=1$.
\subsection{Case I: $x_k$=0}
In this case, the weak user in cell $k$ is served directly by the VLC AP $k$.
Define the variables $\Psi_{s,k}$ and $\Psi_{w,k}$ as $\Psi_{s,k}=\frac{c\nu^2\rho^2h_{s,k}^2}{Z_{s,k}}$, and  $\Psi_{w,k}=\frac{c\nu^2\rho^2h_{w,k}^2}{Z_{w,k}}$, where $Z_{s,k}=B_v N_v+c\nu^2\rho^2\sum_{q=1,q\neq k}^{N_A} p_qh_{s,q}^2$ and $Z_{w,k}=B_v N_v+c\nu^2\rho^2\sum_{q=1,q\neq k}^{N_A} p_qh_{w,q}^2$. Therefore, the optimization problem (\ref{SRT1}) can be written as
\begin{subequations}
\label{SR2}
\begin{eqnarray}
\nonumber
&\displaystyle\max_{P_{s,k}}&  \frac{B_v}{2}\log_2\left( 1+\Psi_{s,k}P_{s,k}\right)\\
&& \ \ \ \ + \frac{B_v}{2}\log_2\left( 1+\frac{p_k-P_{s,k}}{1/\Psi_{w,k}+ P_{s,k}}\right)\\
\label{SR2b}
&\text{s.t.}& R_{w,k}\geq R_{th},\\
\label{SR2c}
&& R_{s,k}\geq R_{th},\\
\label{SR2d}
&&0\leq P_{s,k}\leq \frac{1}{2}p_k.
\end{eqnarray}
\end{subequations}
 Because $P_{w,k}=p_k-P_{s,k}$, constraint (\ref{SR2d}) implies that $0\leq P_{s,k}\leq P_{w,k}$ and $P_{s,k}+ P_{w,k}=p_k$.
\begin{lemma}
 Define the variables $A_{s,k}$ and $C_{w,k}$ as $A_{s,k}=\max(0,\frac{2^{\frac{2R_{th}}{B_v}}-1}{\Psi_{s,k}})$ and $C_{w,k}=\min(\frac{1}{2}p_k,\frac{1+\Psi_{w,k}p_k}{\Psi_{w,k}2^{\frac{2R_{th}}{B_v}}}-\frac{1}{\Psi_{w,k}})$,  the optimal value of $P_{s,k}$, when $x_k=0$, is given by $P_{s,k}^*=P_{s,k,0}$, where
 \begin{equation}
\label{Pjk0}
P_{s,k,0}= \begin{cases} A_{s,k}, & $if  $\Psi_{s,k}<\Psi_{w,k}, \\
C_{w,k}, & $otherwise$.
\end{cases}
\end{equation}
\end{lemma}
\begin{proof}
Based on the above definitions of $A_{s,k}$ and $C_{w,k}$, the constraints in problem (\ref{SR2}) can be rewritten as
\begin{equation}
\label{SRR2b}
A_{s,k} \leq P_{s,k} \leq C_{w,k}.
\end{equation}
%\begin{subequations}
%\label{SRR2}
%\begin{eqnarray}
%\nonumber
%&\displaystyle\max_{P_{s,k}}&  \frac{B_v}{2}\log_2\left( 1+\Psi_{s,k}P_{s,k}\right)\\
%&& \ \ \ \ + \frac{B_v}{2}\log_2\left( 1+\frac{p_k-P_{s,k}}{1/\Psi_{w,k}+ P_{s,k}}\right)\\
%\label{SRR2b}
%&\text{s.t.}& A_{s,k} \leq P_{s,k} \leq C_{w,k}.
%\end{eqnarray}
%\end{subequations}
The derivative of the utility function in (\ref{SR2}) can be written as:
\begin{equation}
\label{dpj2}
\Scale[0.98]{\frac{d}{dP_{s,k}}(R_{s,k}+R_{w,k})= \frac{B_v}{2}\left(\frac{1}{1/\Psi_{s,k}+P_{s,k}} -\frac{1}{1/\Psi_{w,k}+P_{s,k}} \right)}.
\end{equation}
Equation (\ref{dpj2}) implies that the objective function in (\ref{SR2}) is either increasing in $P_{s,k}$ if $\Psi_{s,k}>\Psi_{w,k}$, or decreasing if $\Psi_{s,k}<\Psi_{w,k}$. This means that the optimal value of $P_{s,k}$ is either the minimum bound or the maximum bound of constraint (\ref{SRR2b}). From the above, we conclude that the optimal value of $P_{s,k}$, when $x_k=0$, is given by  $P_{s,k}=P_{s,k,0}$, where $P_{s,k,0}$ is given by (\ref{Pjk0}).
\end{proof}
\subsection{Case II: $x_k$=1}
In this case, the weak user in cell $k$ is served by the strong user through the hybrid VLC/RF link. Hence, problem (\ref{SRT1}) can be written as follows
 \begin{subequations}
\label{SR3}
\begin{eqnarray}
&\displaystyle\max_{P_{s,k}}&  R_{s,k}+ \min \left(R_{w\rightarrow s,k},R_{w,s,k}^{RF} \right)\\
\label{SR3b}
&\text{s.t.}& R_{s,k}\geq R_{th},\\
\label{SR3c}
&& \min(R_{s\rightarrow w,k},R_{w,s,k}^{RF})\geq R_{th},\\
\label{SR3d}
&&0\leq P_{s,k}\leq \frac{1}{2}p_k.
\end{eqnarray}
\end{subequations}
%where $\Psi_{rf,i,j}=\frac{h_{r,i,j}^2}{B_{rf}N_{rf}}$.
%Lets denote the objective function in (\ref{SR3}) as
%\begin{equation}
%F=R_{k,j}+\min(R_{k,i}^{RF}, R_{k,j\rightarrow i})
%\end{equation}
\begin{lemma}
Define the variables $\bar{A}_{s,k}$ and $B_{s,k}$ as $\bar{A}_{s,k}=\max(A_{s,k},\frac{1+\Psi_{s,k}p_k-2^{R_{w,s,k}^{RF}/B_v}}{\Psi_{s,k}2^{R_{w,s,k}^{RF}/B_v}})$, and $B_{s,k}=\min(\frac{1}{2}p_k,\frac{1+\Psi_{s,k}p_k}{\Psi_{s,k}2^{\frac{2R_{th}}{B_v}}}-\frac{1}{\Psi_{s,k}})$, the optimal power allocation of problem (\ref{SR3}) is given by $P_{s,k}^*=P_{s,k,1}$, where $P_{s,k,1}$ is any value within the interval $(\bar{A}_{s,k},B_{s,k})$ (i.e., $\bar{A}_{s,k}\leq P_{s,k,1}\leq B_{s,k}$).
\end{lemma}
\begin{proof}
In problem (\ref{SR3}), it can be seen that $R_{w,s,k}^{RF}$ is a fixed function of $P_{s,k}$. It can also be observed that $R_{w\rightarrow s,k}$ is a decreasing function of $P_{s,k}$, and that $R_{s,k}$ is an increasing function of $P_{s,k}$. This means that the optimal $P_{s,k}$ must satisfy $R_{w\rightarrow s,k}\leq R_{w,s,k}^{RF}$ since, otherwise, we can increase $P_{s,k}$, which increases the objective function without violating the constraints.
 Hence, we replace the term $\min \left(R_{w\rightarrow s,k},R_{w,s,k}^{RF} \right)$ by $R_{w\rightarrow s,k}$ in problem (\ref{SR3}), and add a constraint $R_{w\rightarrow s,k}\leq R_{w,s,k}^{RF}$ instead.  Thus, problem (\ref{SR3}) can be rewritten as:
 \begin{subequations}
\label{SR4}
\begin{eqnarray}
&\displaystyle\max_{P_{s,k}}&  R_{s,k}+  R_{w\rightarrow s,k}\\
\label{SR4b}
&\text{s.t.}& R_{s,k}\geq R_{th},\\
\label{SR4c}
&& R_{w\rightarrow s,k}\geq R_{th},\\
\label{SR4dd}
&&R_{w\rightarrow s,k}\leq R_{w,s,k}^{RF},\\
\label{SR4d}
&&0\leq P_{s,k}\leq \frac{1}{2}p_k.
\end{eqnarray}
\end{subequations}
Take the derivative of the utility in (\ref{SR4}), we get:
\begin{equation}
\label{dpj3}
\Scale[0.9]{\frac{d}{dP_{s,k}}(R_{s,k}+R_{w\rightarrow s,k})=\frac{B_v}{2}\left(\frac{1}{1/\Psi_{s,k}+P_{s,k}} -\frac{1}{1/\Psi_{s,k}+P_{s,k}} \right)}.
\end{equation}
It can be readily seen that equation (\ref{dpj3}) is equal to zero, which implies that the objective function in (\ref{SR4}) is constant with respect to $P_{s,k}$. Any feasible $P_{s,k}$ can, therefore, be conveniently chosen, i.e., such a choice does not affect the optimal solution of (\ref{SR4}).
 Constraints (\ref{SR4b})-(\ref{SR4d}) can be rewritten as follows
\begin{equation}
\label{cons}
\bar{A}_{s,k}\leq P_{s,k} \leq B_{s,k},
\end{equation}
Thus, the optimal solution of problem (\ref{SR4}) is given by $P_{s,k}=P_{s,k,1}$, where $P_{s,k,1}$ can be chosen conveniently from the feasible set: $\bar{A}_{s,k}\leq P_{s,k,1}\leq B_{s,k}$.
\end{proof}

%For the sake of the fairness, we allocate the powers so that the data rates of the users are equal. In other words, we find the powers that achieve $R_{s,k}=R_{w\rightarrow s,k}$, which  maximizes the fairness and the sum-rate at the same time, if permitted by constraints (\ref{SR4b})-(\ref{SR4d}). By solving equation $R_{s,k}=R_{s\rightarrow w,k}$, we obtain $P_{s,k}=\eta_{s,k}$, where $\eta_{s,k}$ is given by
%\begin{equation}
%\eta_{s,k}=\frac{\sqrt{1+\Psi_{s,k}p_k}-1}{\Psi_{s,k}}.
%\end{equation}

In our simulation results, we choose to set $P_{s,k,1}$ to jointly achieve (\ref{cons}) and maximize the system fairness simultaneously. Hence, $P_{s,k,1}$ is expressed as follows
\begin{equation}
\label{Pjk1}
P_{s,k,1}= \begin{cases} \bar{A}_{s,k}, & $if  $\eta_{s,k}<A_{s,k}, \\
\eta_{s,k}, & $if $ A_{s,k}\leq \eta_{s,k} \leq B_{s,k}\\
B_{s,k}, & $if $ \eta_{s,k} > B_{s,k},
\end{cases}
\end{equation}
where $\eta_{s,k}=\frac{\sqrt{1+\Psi_{s,k}p_k}-1}{\Psi_{s,k}}$ is the value that achieves equal rate for both the strong and weak users at the cell $k$ (i.e., $\eta$ is the root of the equation resulting from equating (\ref{Rsk}) to (\ref{Rwk}), where $P_{s,k}=\eta$ and $P_{w,k}=p_k-\eta$).
%achieves that $R_{s,k}=R_{w,k}$, where $R_{s,k}$ is given by
At this stage, we are capable of finding closed-form solutions for the joint users' power and link selection problems for every cell $k$, both for the cases when $x_k=0$ or $x_k=1$. The chosen solution of every cell $k$ is then the pair $P_{s,k}$ and $x_k$ that maximizes the utility function in (\ref{SRT1})). Algorithm \ref{Alg1} summarizes these procedures.
\begin{algorithm}
 \caption{Find the vectors $\mathbf{x}$ and $\mathbf{P}$, given $\mathbf{p}$}
 \label{Alg1}
 \begin{enumerate}
 \item For each AP $k=1,\cdots,N$;
 \item Find $P_{s,k,0}$ and $P_{s,k,1}$, using (\ref{Pjk0}) and (\ref{Pjk1}), respectively, i.e., when $x_k=0$ and $x_k=1$.
 \item Choose the pair $P_{s,k}$ and $x_k$ that maximizes the objective function in (\ref{SRT1}).
 \end{enumerate}
 \end{algorithm}
\subsection{Optimizing the AP Transmit Power}
Recall that the above per AP-formulation (\ref{SRT1}) of the original problem (\ref{SRT}) only holds for a fixed AP transmit power $p_k$. The papers now solves for the transmit power $\mathbf{p}$ that maximizes problem (\ref{SRT}). The proposed solution is iterative in nature, as each AP shares its instantaneous transmit power and users' channels with other APs. The idea is that the AP $k$ uses the shared information of the users' channels and the AP transmit powers to calculate the objective function in (\ref{SRT}) in order to find a local optimal $p_k$, using the golden section method. We define the minimum transmit power that can achieve constraints (\ref{SRT1b})-(\ref{SRT1h}) as $P_{min,k}=\min\big(P_{max}, \frac{A^2-A}{\Psi_{s,k}}+\frac{A-1}{\Psi_{w,k}}\big)$, where $A=2^{2R_{th}/B_v}$ and $P_{max}=(I_H-b)^2$. Algorithm \ref{Alg2} can be used to find a joint solution of the vectors $\mathbf{P}$, $\mathbf{p}$, and $\mathbf{x}$.
\begin{algorithm}
 \caption{Find the vectors $\mathbf{x}$, $\mathbf{P}$, and $\mathbf{p}$}
 \label{Alg2}
 \begin{enumerate}
\item[1] Run Algorithm \ref{Alg1}, when $p_k=P_{max} \forall k=1,\ldots N$.
\item[2] Starting from the highest interfering AP to the lowest one, for each AP,
\item[3] Assign $m=P_k$, $n=P_{min,k}$, $\theta=1.618$
\item[4] While $m-n\leq \epsilon$
\item[5] Implement Algorithm \ref{Alg1}, when $p_k=a$, where $a=(\theta-1)n+(2-\theta)m$ and set the resulted objective function in (\ref{SRT}) as $R_a$
\item[5] Implement Algorithm \ref{Alg1}, when $p_k=b$, where $b=(2-\theta)n+(\theta-1)m$ and set the resulted objective function in (\ref{SRT}) as $R_b$
\item[6] If $R_a>R_b$, set $m=a$, else set $n=a$.
\item[7] end while
\item[8] Set $p_k=(m+n)/2$
\end{enumerate}
 \end{algorithm}
Algorithm \ref{Alg2} first finds the joint users' power and link selection with the worst interference case (i.e., all APs transmit with total available power). Under the assumption that the APs share their channels and transmit power information, Algorithm \ref{Alg2} circulates over the APs to find a local optimal solution of each AP's transmit power and then update the shared information. In other words, at the $k$th iteration, Algorithm \ref{Alg2} uses the shared information to find the transmit power of AP $k$ (using the golden section method), updates the shared information, moves on to the AP $k+1$ to do the same process, and so on. This process can be repeated overall the APs several times.  At each iteration, it can be seen that the updated transmit power of the considered AP either increases the objective function in (\ref{SRT}) or makes it fixed, which proves the convergence of Algorithm \ref{Alg2}. In the simulation results, we numerically prove the convergence of Algorithm \ref{Alg2} as shown in Fig. \ref{Itr_SR}.   
 
 To analyse the computational complexity of Algorithm \ref{Alg2}, it can be seen that at each iteration, Algorithm \ref{Alg1} is implemented twice. Algorithm \ref{Alg1} just solves three equations, which are (\ref{Pjk0}), (\ref{Pjk1}), and the objective function in (\ref{SRT}). Therefore, if the number of iterations is $I$, the complexity of Algorithm \ref{Alg2} is of order $O(6I)$, where in each iteration, one AP updates its transmit power. We should note that the number of iterations $I$ could be more than or equal the number of APs $N$ due to that Algorithm \ref{Alg2} must circulate over all the APs at least one round.  
%\begin{table}[!t]
%\centering
%\caption{Simulation Parameters}
%\label{table1}
%\begin{tabular}{|p{.27\textwidth} | p{.17\textwidth} |}
%\hline
%  Name of the Parameter & Value of the Parameter\\
% \hline
%
%  Maximum bandwidth of VLC AP, $B$ & $20$ MHz  \\
%
%  The physical area of a PD, $A_{p}$ & $1$\ cm$^2$ \\
%   Half-intensity radiation angle, $\theta_{1/2}$ & $60^o$\  \\
%  Gain of optical filter, $g_{of}$ & $1$  \\
%  Refractive index, $n$ & 1.5 \\
%   Efficiency of converting optical to electric, $\rho$& $0.53$ [A/W]\\
%   Noise power spectral density of LiFi, $N_0$ & $10^{-21}$\ A$^2$/Hz  \\
% %The PD responsivity, $R_{pd}$& 0.53\\
%  Maximum input bias current, $I_H$ & $600$ mA  \\
%
%  Minimum input bias current, $I_L$ & $400$ mA  \\
%  Fill factor, $f$ &0.75\\
%  LEDs' power, $\nu$ & 10 W/A\\
%
%  Thermal voltage, $V_t$ & 25 mV \\
%
%  Dark saturation current of the PD, $I_0$ & $10^{-10}$ A\\
%  LED height, &$3$ m\\
%  User height & $0.85$\\
%  Monte-Carlo for user distribution,  & 200 different user distributions\\
%
%
%  \hline
%
%  RF    \\
%  \hline
%  The  breakpoint distance & 5 m\\
%  RF bandwidth assigned for a user & 8 MHz\\
%  Central carrier frequency  & 2.4 GHz\\
%  Angle of arrival/departure of LoS & 45$^o$\\
%  Shadow fading standard deviation (before the breakpoint) & 3 dB  \\
%  Shadow fading standard deviation (after the breakpoint) & 5 dB  \\
%  Spectral density of the noise power & -174 dBm/Hz\\
%  \hline
%\end{tabular}

% \end{table}

\begin{figure} [!t]
\centering
\includegraphics[width=3.5in]{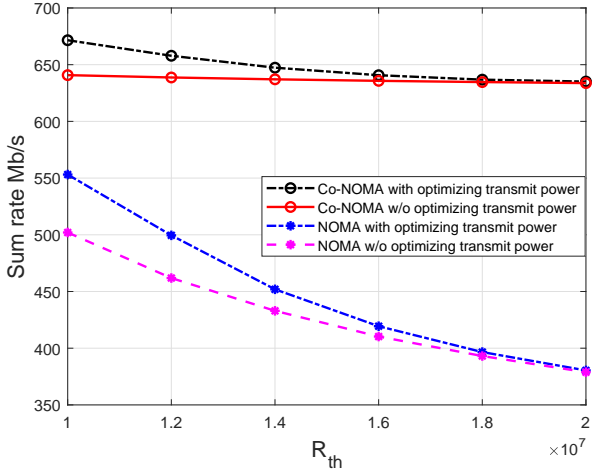}
\caption{The effect of $R_{th}$ on sum-rate.}
\label{R_SR}
\end{figure}
 \section{Simulation Results}
 \label{SResult}
%\begin{figure} [!t]
%    \centering
%    \label{R_SR}
%  \subfloat[Sum-rate versus $R_{th}$.]{
%       \includegraphics[width=0.52\linewidth, height=4.8cm]{R_SR2.eps}}
%  \subfloat[Fairness versus $R_{th}$.]{
%        \includegraphics[width=0.52\linewidth, height=4.8cm]{R_F2.eps}}
%    \label{R_F}
%    \caption{The effect of $R_{th}$ on sum-rate and fairness.}
%    \label{Rth5}
%\end{figure}

 In this section, we assess the performance of the proposed algorithms in a Co-NOMA hybrid VLC/RF system. In the simulations, we illustrate how changing the required QoS and increasing the average interference would affect the sum-rate and fairness. Note that Jain's fairness index is used to measure the system fairness \cite{jain1984quantitative}.
%$\frac{\left(\sum_{k=1}^N \left( R_{s,k}+(1-x_{k})R_{w,k}+x_{k}R_{w,k}^{HL}\right)\right)^2}{2N\left(\sum_{k=1}^N \left( R_{s,k}^2+(1-x_{k})R_{w,k}^2+x_{k}(R_{w,k}^{HL})^2\right)\right)}$.
The number of the APs in the ceiling is set to $16$, and the separation distance between them is set to $2.5$ m. Simulation parameters related to the channel values and transmit powers are chosen similar to Table II in reference \cite{CoNOMA}. We evaluate the proposed solutions through Mote-Carlo simulations, where each point in the following figures is the average of 200 different users' distributions within the restrictions illustrated in the System Model Section.

\begin{figure} [!t]
\centering
\includegraphics[width=3.3in]{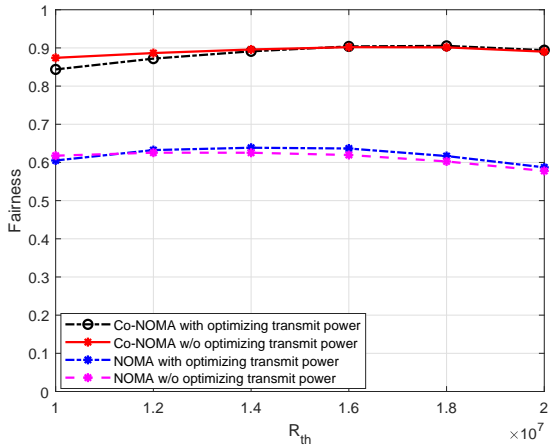}
\caption{The effect of $R_{th}$ on fairness.}
\label{R_F}
\end{figure}

%Fig. 3 shows the effect of $R_{th}$ on the sum rate and fairness of the proposed system. It can be seen that the scheme Co-NOMA outperforms the traditional scheme NOMA in terms of sum rate and fairness and when the transmit power is  either optimized or not. The significant improvement in performance in Co-NOMA comes from the fact that each weak user has two options to select either to be served through the hybrid RF/VLC link or the direct link, while in NOMA scheme each weak user has only one option which is to be served through the direct VLC link.

Fig. \ref{R_SR} plots the system sum-rate versus the target data rate $R_{th}$, by assuming that all users in the system require the same $R_{th}$. The figure shows that increasing the target data rate decreases the system sum-rate, especially under the non-cooperative NOMA scheme. This is the case because increasing $R_{th}$ at the weak user decreases the strong user power, which decreases the overall sum-rate. Fig. \ref{R_SR} particularly shows that the Co-NOMA scheme outperforms the non-cooperative NOMA scheme in terms of sum-rate, both with and without the power optimization step. The significant improvement in the performance in Co-NOMA comes from the fact that each weak user can select between the hybrid RF/VLC and the direct VLC links, while in NOMA scheme, each weak user can only be served through the direct VLC link.

To illustrate the system fairness performance of the proposed scheme, Fig. \ref{R_F} plots Jain's index versus the target data rate $R_{th}$. The figure shows that the fairness of Co-NOMA is much better than the fairness of NOMA systems. This is particularly the case because the weak users in NOMA suffer from inter-cell interference, while the edge users in Co-NOMA can be served through the center (strong) users, which are in relatively good channel conditions.

%the effect of $R_{th}$ on the fairness is negligible because the range of changing $R_{th}$ is small and we keep it in this range to guarantee that the successive interference cancellation constraint achieved.

%Fig. \ref{Alpha_SR} shows how increasing the average received interference at weak users would affect the system sum rate and fairness. This can be implemented by increasing the value of $\alpha$, where increasing $\alpha$ increases the probability that the weak user be at the edge of the cell (as shown in Fig. \ref{SM2}), which increases the average received interference at system weak users. The figure shows that increasing the interference at the weak users approximately doesn't affect the sum rate and the fairness in Co-NOMA system, while it deteriorates the performance of NOMA scheme significantly. This is because increasing the  interference at the direct VLC links of the weak users does not affect the hybrid VLC/RF link.

\begin{figure} [!t]
\centering
\includegraphics[width=3.3in]{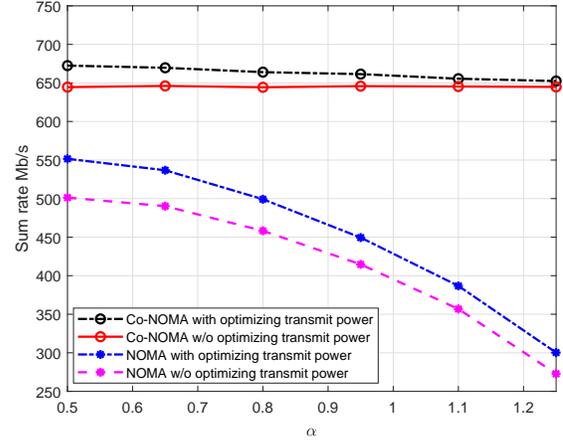}
\caption{The effect of increasing the interference at weak users on sum-rate.}
\label{Alpha_SR}
\end{figure}

\begin{figure} [!t]
\centering
\includegraphics[width=3.3in]{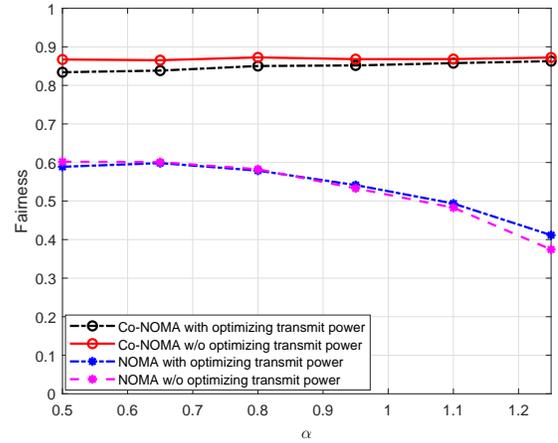}
\caption{The effect of increasing the interference at weak users on the fairness.}
\label{Alpha_F}
\end{figure} 

Fig.~\ref{Alpha_SR} shows how increasing the average received interference at weak users affects the system's sum rate. This is implemented by increasing the value of $\alpha$, i.e., the average interference received by the weak users, while keeping the area of the strong user constant. Such increase in $\alpha$ pushes the distribution area of the weak users to the cell-edge. Fig.~\ref{Alpha_SR} shows that increasing the interference at the weak users does not significantly affect the sum-rate in Co-NOMA, but deteriorates the performance of the non-cooperative NOMA significantly. Such behavior is mainly due to the fact that as the interference increases, the weak users in Co-NOMA migrate from being served through the interfered VLC links to being served through the hybrid VLC/RF links, which highlights the pronounced role of the proposed scheme in mitigating interference, especially in high interference regimes.

Fig. \ref{Alpha_F} also shows the superiority of Co-NOMA over the non-cooperative NOMA in terms of fairness. As the interference at the weak users increases, the fairness of NOMA decreases because the interference only increases at the weak users (not at the strong users), which increases the difference between the data rates of the weak users and the strong users. 
%\begin{figure}[!t]
%    \centering
%    \label{Alpha_SR}
%  \subfloat[$\alpha$ versus sum-rate]{
%       \includegraphics[width=0.5\linewidth, height=4.8cm]{Alpha_SR2.eps}}
%  \subfloat[$\alpha$ versus fairness]{
%        \includegraphics[width=0.5\linewidth, height=4.8cm]{Alpha_F2.eps}}
%    \label{Alpha_F}
%    \caption{The effect of increasing the interference at weak users on the sum-rate and fairness}
%    \label{Alphak}
%\end{figure}

\begin{figure}[!h]
\centering
\includegraphics[width=3.3in]{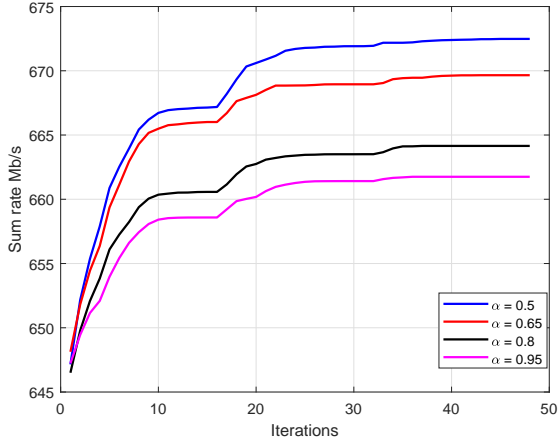}
\caption{Convergence of Algorithm \ref{Alg2}.}
\label{Itr_SR}
\end{figure}
Finally, Fig. \ref{Itr_SR} shows the convergence of Algorithm \ref{Alg2} by plotting the the sum-rate versus the number of iterations for different levels of average interference received by the weak users. The figure demonstrates that Algorithm \ref{Alg2} converges in few iterations, e.g., in 25 iterations when $\alpha=0.95$. In each iteration, one AP uses Algorithm \ref{Alg1} to find the users' power and link selection vector, and then uses these solutions to find a local optimal solution of the total transmit power of that AP. Algorithm \ref{Alg2} proposes to go over all APs several rounds, and this is why the number of iterations scales the number of APs.  
\section{Conclusion}
VLC are expected to play a major role in meeting the ambitious metrics of future wireless systems. This paper applies the Co-NOMA scheme in a multicell VLC network, and maximizes the sum-rate by determining the power and link selection vectors under power and QoS constraints. The paper solves such a non-convex problem by first finding closed form solutions of the joint users' powers and link selection, for a fixed AP power. The APs' transmit powers are then solved in an outer loop using the golden section method. Simulation results show how the proposed scheme outperforms non-cooperative NOMA scheme in terms of sum-rate and fairness. In addition, simulation results show that Algorithm \ref{Alg2} is convergent and improves the system performance in terms of sum-rate, especially in the non-cooperative NOMA scheme. 
\newpage

\bibliography{mylib}
\bibliographystyle{IEEEtran}

\end{document}